\documentclass[12pt,onecolumn,journal]{IEEEtran}

\usepackage[final]{graphicx}
\usepackage[reqno]{amsmath}
\usepackage{amssymb}
\usepackage{cite}

\newfont{\bbb}{msbm10 scaled 500}

\newfont{\bb}{msbm10 scaled 1100}

\newcommand{\cv}{{\bf c}}

\newcommand{\ov}{{\bf o}}

\newcommand{\qv}{{\bf q}}

\newcommand{\sv}{{\bf s}}

\newcommand{\xv}{{\bf x}}
\newcommand{\yv}{{\bf y}}

\newcommand{\Cm}{{\bf C}}

\newcommand{\Om}{{\bf O}}

\newcommand{\Qm}{{\bf Q}}

\newcommand{\Sm}{{\bf S}}

\newcommand{\Vm}{{\bf V}}
\newcommand{\Xm}{{\bf X}}
\newcommand{\Ym}{{\bf Y}}

\newcommand{\Ac}{{\cal A}}

\newcommand{\Cc}{{\cal C}}

\newcommand{\Nc}{{\cal N}}
\newcommand{\Oc}{{\cal O}}
\newcommand{\Pc}{{\cal P}}
\newcommand{\Qc}{{\cal Q}}
\newcommand{\Rc}{{\cal R}}
\newcommand{\Sc}{{\cal S}}

\newcommand{\Wc}{{\cal W}}

\newcommand{\Xc}{{\cal X}}
\newcommand{\Yc}{{\cal Y}}

\newtheorem{theorem}{Theorem}
\newtheorem{corollary}[theorem]{Corollary}
\newtheorem{lemma}[theorem]{Lemma}

\IEEEoverridecommandlockouts

\author{
 \IEEEauthorblockN{O. Ozan Koyluoglu 
 and Hesham El Gamal
 \thanks{O. Ozan Koyluoglu and Hesham El Gamal are with the 
 Department of Electrical and Computer
 Engineering, The Ohio State University, Columbus, OH 43210, USA.
 (e-mail: \{koyluogo,helgamal\}@ece.osu.edu).}
 \thanks{This work is submitted to IEEE Transactions on Information Theory
 in May 2009, and revised in September 2010.}
 \thanks{This research was supported in
 part by the National Science Foundation under Grant CCF-07-28762,
 and in part by Los Alamos National Laboratory and Qatar National
 Research Fund.
 The first author is supported in part by the Presidential
 Fellowship Award of the Ohio State University.}
 }}
\title{Cooperative encoding for secrecy in interference channels}

\begin{document}
\maketitle


\begin{abstract}

This paper investigates the fundamental performance limits of the two-user
interference channel in the presence of an external eavesdropper. In this
setting, we construct an inner bound, to the secrecy capacity region, based
on the idea of cooperative encoding in which the two
users {\em cooperatively} design their randomized codebooks and {\em jointly}
optimize their channel prefixing distributions. Our achievability scheme
also utilizes message-splitting in order to allow for partial decoding of
the interference at the non-intended receiver.
Outer bounds are then derived and used to establish the optimality of the
proposed scheme in certain cases. In the Gaussian case, the previously
proposed cooperative jamming and noise-forwarding techniques are shown
to be special cases of our proposed approach. Overall, our results provide
structural insights on how the interference can be {\em exploited} to
increase the secrecy capacity of wireless networks.
\end{abstract}




\section{Introduction}
\label{sec:Introduction}

Without the secrecy constraint, the interference channel has been
investigated extensively in the literature. The best known achievable region
was obtained in~\cite{Han:A81} and was recently simplified in~\cite{Chong:On08}.
However, except for some special cases
(e.g., \cite{Sato:The81,Costa:The87,Shang:A,Annapureddy:Gaussian,Motahari:Capacity}),
characterizing the capacity region of the two-user Gaussian interference
channel remains as an open problem. On the other hand, recent attempts have
shed light on the fundamental limits of
the interference channels with confidential
messages~\cite{Liang:Cognitive07,Liu:Discrete08,Koyluoglu:On08,Koyluoglu:Interference}.
Nonetheless, the external eavesdropper scenario, the model considered here, has
not been addressed adequately in the literature. In fact, to the best of
our knowledge, the only relevant work is the recent study on the achievable
secure degrees of freedom (DoF) of the $K$-user Gaussian interference
channels under a frequency selective
fading model~\cite{Koyluoglu:On08,Koyluoglu:Interference}.

This work develops a general approach for cooperative encoding
for the (discrete) two-user memoryless interference
channels operated in the presence of a passive eavesdropper. The proposed
scheme allows for cooperation in two distinct ways:

$1$) The two users jointly optimize their randomized
(two-dimensional) codebooks~\cite{Wyner:The75}, and

$2$) The two users jointly introduce randomness in the transmitted
signals, to confuse the eavesdropper, via a cooperative channel
prefixing approach~\cite{Csiszar:Broadcast78}.

Remarkably, the two methods, respectively, are helpful in adding
\emph{decodable} and \emph{undecodable} randomness to the channel.
The proposed scheme also utilizes message-splitting and
partial decoding to enlarge the achievable secrecy rate
region~\cite{Han:A81}. We then derive
outer bounds to the secrecy capacity region and use them to
establish the optimality of the proposed scheme for some classes
of channels. In addition, we argue that some coding techniques
for the secure discrete multiple access channel and relay-eavesdropper
channel can be obtained as special cases of our cooperative
encoding approach.

Recently, noise forwarding (or jamming) has been shown to
enhance the achievable secrecy rate region of several
Gaussian multi-user channels (e.g.,~\cite{Negi:Secret05}
,~\cite{Tekin:The08}). The basic idea is to allow each transmitter
to allocate only a fraction of the available power for its randomized
codebook and use the rest for the generation of independent
noise samples. The superposition of the two signals is then
transmitted. With the appropriate power allocation policy, one
can ensure that the jamming signal results in maximal ambiguity at
the eavesdropper while incuring only a minimal loss in the achievable
rate at the legitimate receiver(s). Our work reveals the fact that
this noise injection technique can be obtained as a manifestation
of the cooperative channel prefixing approach. Based on this
observation, we obtain a larger achievable region for the
secrecy rate in the Gaussian multiple access channel.

\IEEEpubidadjcol

The rest of the paper is organized as follows.
Section~\ref{sec:II} is devoted to the discrete memoryless
scenario where the main results of the paper are derived
and few special cases are analyzed. The analysis for the
Gaussian channels, along with numerical results in selected
scenarios, are given in Section~\ref{sec:III}.
Finally, we offer some concluding remarks in Section~\ref{sec:IV}.
The proofs are collected in the appendices to enhance
the flow of the paper.


\section{Security for the Discrete Memoryless Interference Channel}
\label{sec:II}


\subsection{System Model and Notations}
Throughout this paper, vectors are denoted as $\xv^i=\{x(1),\cdots,x(i)\}$, where
we omit the subscript $i$ if $i=n$, i.e., $\xv=\{x(1),\cdots,x(n)\}$.
Random variables are denoted with capital letters $X$, which
are defined over sets denoted by the calligraphic letters $\Xc$,
and random vectors are denoted as bold-capital letters $\Xm^i$.
Again, we drop the subscript $i$ for $\Xm=\{X(1),\cdots,X(n)\}$.
We define, $[x]^+\triangleq\max\{0,x\}$,
$\bar{\alpha}\triangleq 1-\alpha$, and
$\gamma(x)\triangleq\frac{1}{2}\log_2(1+x)$.
The delta function $\delta(x)$ is defined as
$\delta(x)=1$, if $x=0$; $\delta(x)=0$, if $x\neq0$.
Also, we use the following shorthand for
probability distributions: $p(x)\triangleq p(X=x)$,
$p(x|y)\triangleq p(X=x|Y=y)$. The same notation will be used
for joint distributions. For any randomized codebook, we construct
$2^{nR}$ rows (corresponding to message indices) and
$2^{nR^x}$ codewords per row (corresponding to randomization indices),
where we refer to $R$ as the secrecy rate and $R^x$ as the randomization
rate. (Please refer to Fig.~\ref{fig:1}.)
Finally, for a given set $\Sc$,
$R_{\Sc} \triangleq \sum\limits_{i\in\Sc} R_i$ for secrecy rates
and $R_{\Sc}^x \triangleq \sum\limits_{i\in\Sc} R_i^x$ for the
randomization rates.

Our discrete memoryless two-user interference channel with
an (external) eavesdropper (IC-E) is denoted by
$$(\Xc_1 \times \Xc_2, p(y_1,y_2,y_e|x_1,x_2),
\Yc_1 \times \Yc_2 \times \Yc_e),$$ for some finite sets
$\Xc_1, \Xc_2, \Yc_1, \Yc_2, \Yc_e$ (see Fig.~\ref{fig:2}).
Here the symbols $(x_1,x_2)\in \Xc_1 \times \Xc_2$ are the
channel inputs and the symbols $(y_1,y_2,y_e)\in \Yc_1 \times \Yc_2
\times \Yc_e$ are the channel outputs observed at the
decoder $1$, decoder $2$, and at the eavesdropper, respectively.
The channel is memoryless and
time-invariant:
$$p(y_1(i),y_2(i),y_e(i)|\xv_1^i,\xv_2^i,
\yv_1^{i-1},\yv_2^{i-1},\yv_e^{i-1})
= p(y_1(i),y_2(i),y_e(i)|x_1(i),x_2(i)).$$

We assume that each transmitter $k\in\{1,2\}$ has a secret message
$W_k$ which is to be transmitted to the respective receiver in $n$
channel uses and to be secured from the eavesdropper.
In this setting,
an $(n,M_1,M_2,P_{e,1},P_{e,2})$ secret codebook
has the following components:

$1$) The secret message set $\Wc_k=\{1,...,M_k\}$; $k=1,2$.

$2$) A stochastic encoding function $f_k(.)$ at transmitter $k$
which maps the secret messages to the transmitted symbols:
$f_k:w_k\to \Xm_k$ for each $w_k\in\Wc_k$; $k=1,2$.

$3$) Decoding function $\phi_k(.)$ at receiver $k$
which maps the received symbols to an estimate of the message:
$\phi_k(\Ym_k)=\hat{w}_k$; $k=1,2$.

\begin{figure}[t] 
    \centering
    \includegraphics[width=0.3\columnwidth]{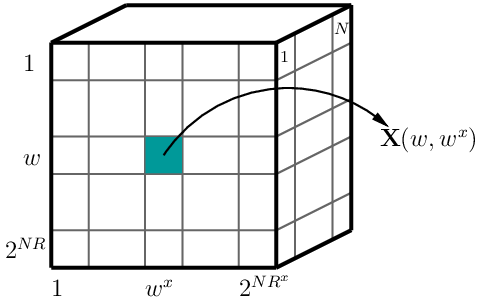}
    \caption{Randomized codebook of Wyner. For a given message index,
    a column is randomly chosen, and the corresponding channel
    input is transmitted. Designing the number of columns properly
    is the key to prove that the security constraint is met.
    We refer to this codebook as
    randomized (two-dimensional) codebook.}\label{fig:1}
\end{figure}

\begin{figure}[t] 
    \centering
    \includegraphics[width=0.7\columnwidth]{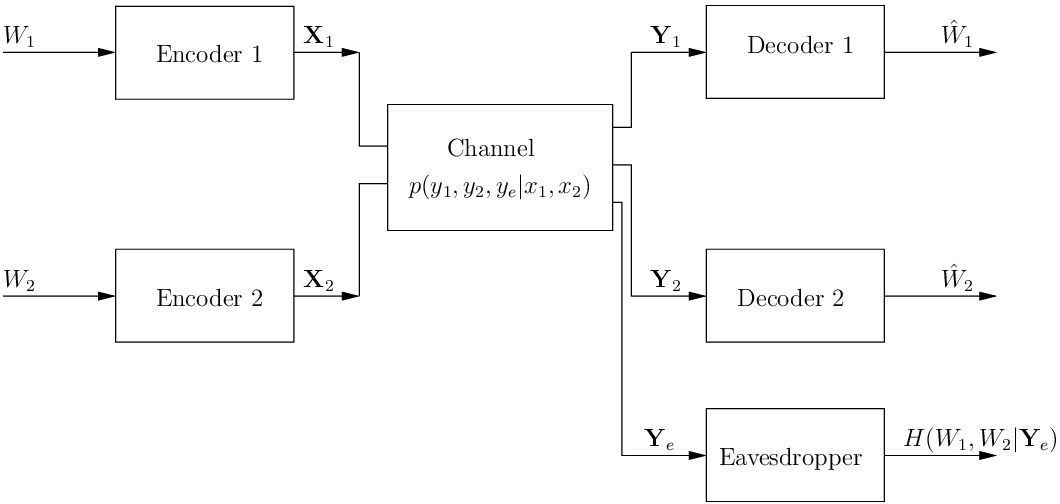}
    \caption{The discrete memoryless interference channel
    with an eavesdropper (IC-E).}\label{fig:2}
\end{figure}

The reliability of transmission is measured by the following
probabilities of error
\begin{eqnarray}
P_{e,k}=\frac{1}{M_1 M_2}\sum\limits_{(w_1,w_2)\in \Wc_1\times\Wc_2}
\textrm{Pr} \bigg\{\phi_k(\Ym_k)\neq w_k | (w_1,w_2) \textrm{ is sent}\bigg\},\nonumber
\end{eqnarray}
for $k=1,2$. The secrecy is measured by the equivocation
rate
$$\frac{1}{n}H\left(W_1,W_2|\Ym_e\right).$$

We say that the rate tuple $(R_1,R_2)$ is achievable for the IC-E
if, for any given $\epsilon > 0$, there exists an
$(n,M_1,M_2,P_{e,1},P_{e,2})$ secret codebook such that,
\begin{eqnarray}\label{eq:secrecy}
\frac{1}{n}\log(M_1) &=& R_1 \nonumber\\
\frac{1}{n}\log(M_2) &=& R_2 , \nonumber\\
\max\{P_{e,1},P_{e,2}\}&\leq& \epsilon \nonumber,
\end{eqnarray} and
\begin{eqnarray}
R_1+R_2 -
\frac{1}{n}H\left(W_1,W_2|\Ym_e\right) &\leq& \epsilon
\end{eqnarray}
for sufficiently large $n$. The secrecy capacity region is the closure of
the set of all achievable rate pairs $(R_1,R_2)$ and is denoted as
$\mathbb{C}^{\textrm{IC-E}}$. Finally, we note that
the secrecy requirement imposed on the full message set implies the
secrecy of individual messages. In other words,
$\frac{1}{n}I(W_1,W_2;\Ym_e)\leq\epsilon$ implies
$\frac{1}{n}I(W_k;\Ym_e)\leq\epsilon$ for $k=1,2$.


\subsection{Inner Bound}

In this section, we introduce the cooperative encoding
scheme, and derive an inner bound
to $\mathbb{C}^{\textrm{IC-E}}$.
The proposed strategy allows for cooperation in design of the randomized
codebooks, as well as in channel prefixing~\cite{Csiszar:Broadcast78}.
This way, each user will add only
\emph{a sufficient} amount of randomness as the other
user will help to increase the randomness seen by the eavesdropper.
The achievable secrecy rate region using this approach is
stated in the following result.

\begin{theorem}\label{thm:R}

\begin{equation}\Rc^{\textrm{IC-E}}
\triangleq
\textrm{ the closure of }
\left\{
\bigcup\limits_{p\in\Pc} \Rc(p)
\right\}
\subset \mathbb{C}^{\textrm{IC-E}},
\end{equation}
where $\Pc$ denotes the set of all joint distributions of the random variables
$Q$, $C_1$, $S_1$, $O_1$, $C_2$, $S_2$, $O_2$, $X_1$, $X_2$
that factors
as~\footnote{Here $Q, C_1, S_1, O_1, C_2, S_2$, and $O_2$ are
defined on arbitrary finite sets $\Qc$, $\Cc_1$, $\Sc_1$, $\Oc_1$,
$\Cc_2$, $\Sc_2$, and $\Oc_2$, respectively.}

\begin{eqnarray} 
p(q,c_1,s_1,o_1,c_2,s_2,o_2,x_1,x_2)&=&
p(q)p(c_1|q)p(s_1|q)p(o_1|q)
p(c_2|q)p(s_2|q)p(o_2|q) \nonumber\\
&&{}p(x_1|c_1,s_1,o_1)
p(x_2|c_2,s_2,o_2),
\end{eqnarray}
and $\Rc(p)$ is the closure of all $(R_1,R_2)$ satisfying
\begin{eqnarray}
R_1&=&R_{C_1}+R_{S_1}, \nonumber\\
R_2&=&R_{C_2}+R_{S_2}, \nonumber\\
(R_{C_1},R_{C_1}^x,R_{S_1},R_{S_1}^x,R_{C_2},R_{C_2}^x,R_{O_2}^x) &\in& \Rc_1(p), \nonumber\\
(R_{C_2},R_{C_2}^x,R_{S_2},R_{S_2}^x,R_{C_1},R_{C_1}^x,R_{O_1}^x) &\in& \Rc_2(p), \nonumber\\
(R_{C_1}^x,R_{S_1}^x,R_{O_1}^x,R_{C_2}^x,R_{S_2}^x,R_{O_2}^x) &\in& \Rc_e(p), \nonumber
\end{eqnarray}
and
\begin{eqnarray}\label{eq:R}
R_{C_1}\geq 0, R_{C_1}^x\geq 0, R_{S_1}\geq 0, R_{S_1}^x\geq 0, R_{O_1}^x\geq 0,\nonumber\\
R_{C_2}\geq 0, R_{C_2}^x\geq 0, R_{S_2}\geq 0, R_{S_2}^x\geq 0, R_{O_2}^x\geq 0,
\end{eqnarray}
for a given joint distribution $p$. $\Rc_1(p)$ is the set of all tuples
$(R_{C_1},R_{C_1}^x,R_{S_1},R_{S_1}^x,R_{C_2},R_{C_2}^x,R_{O_2}^x)$ satisfying
\begin{eqnarray}
R_{\Sc} + R_{\Sc}^x \leq I(\Sc;Y_1|\Sc^c,Q),\:
\forall \Sc \subset \{C_1,S_1,C_2,O_2\}.
\end{eqnarray}
$\Rc_2(p)$ is the rate region defined by reversing the
indices $1$ and $2$ everywhere in the expression for $\Rc_1(p)$.
$\Rc_e(p)$ is the set of all tuples
$(R_{C_1}^x,R_{S_1}^x,R_{O_1}^x,R_{C_2}^x,R_{S_2}^x,R_{O_2}^x)$
satisfying
\begin{eqnarray}\label{eq:R_e(p)}
R_{\Sc}^x & \leq & I(\Sc;Y_e|\Sc^c,Q),\:
\forall \Sc \subsetneq \{C_1,S_1,O_1,C_2,S_2,O_2\}, \nonumber\\
R_{\Sc}^x & = & I(\Sc;Y_e|Q),\:
\Sc = \{C_1,S_1,O_1,C_2,S_2,O_2\}.
\end{eqnarray}

\end{theorem}
\begin{proof}
We detail the coding scheme here.
The rest of the proof, given in Appendix~\ref{sec:ProofInnerBound},
shows that the proposed coding scheme satisfies both the reliability
and the security constraints.

Fix some $p(q)$, $p(c_1|q)$, $p(s_1|q)$, $p(o_1|q)$, $p(x_1|c_1,s_1,o_1)$,
$p(c_2|q)$, $p(s_2|q)$, $p(o_2|q)$, and $p(x_2|c_2,s_2,o_2)$ for the
channel given by $p(y_1,y_2,y_e|x_1,x_2)$.
Generate a random typical sequence $\qv^n$, where
$p(\qv^n)=\prod\limits_{i=1}^n p(q(i))$ and each entry is chosen
i.i.d. according to $p(q)$. Every node knows the sequence $\qv^n$.

\textbf{Codebook Generation:}

Consider transmitter $k\in\{1,2\}$ that has secret message
$W_k\in\Wc_k=\{1,2,\cdots,M_k\}$, where $M_k=2^{nR_k}$.
We construct each element in the codebook ensemble as follows.
\begin{itemize}
  \item Generate $M_{C_k}M_{C_k}^x=2^{n(R_{C_k}+R_{C_k}^x-\epsilon_1)}$
sequences $\cv_k^n$, each with probability
$p(\cv_k^n|\qv^n)=\prod\limits_{i=1}^n p(c_k(i)|q(i))$,
where $p(c_k(i)|q(i))=p(c_k|q)$ for each $i$. Distribute these into
$M_{C_k}=2^{nR_{C_k}}$  bins, where the bin index is $w_{C_k}$. Each bin has
$M_{C_k}^x=2^{n(R_{C_k}^x-\epsilon_1)}$ codewords,
where we denote the codeword index
as $w_{C_k}^x$. Represent each codeword with these two indices, i.e.,
$\cv_k^n(w_{C_k},w_{C_k}^x)$.
  \item Similarly, generate $M_{S_k}M_{S_k}^x=2^{n(R_{S_k}+R_{S_k}^x-\epsilon_1)}$
sequences $\sv_k^n$, each with probability
$p(\sv_k^n|\qv^n)=\prod\limits_{i=1}^n p(s_k(i)|q(i))$,
where $p(s_k(i)|q(i))=p(s_k|q)$ for each $i$. Distribute these into
$M_{S_k}=2^{nR_{S_k}}$  bins, where the bin index is $w_{S_k}$. Each bin has
$M_{S_k}^x=2^{n(R_{S_k}^x-\epsilon_1)}$ codewords,
where we denote the codeword index as $w_{S_k}^x$.
Represent each codeword with these two indices, i.e.,
$\sv_k^n(w_{S_k},w_{S_k}^x)$.
  \item Finally, generate $M_{O_k}^x=2^{n(R_{O_k}^x-\epsilon_1)}$ sequences $\ov_k^n$, each
with probability $p(\ov_k^n|\qv^n)=\prod\limits_{i=1}^n p(o_k(i)|q(i))$,
where $p(o_k(i)|q(i))=p(o_k|q)$ for each $i$.
Denote each sequence by index $w_{O_k}^x$ and represent each codeword
with this index, i.e., $\ov_k^n(w_{O_k}^x)$.
\end{itemize}
Choose $M_k=M_{C_k} M_{S_k}$, and assign each pair $(w_{C_k},w_{S_k})$
to a secret message $w_k$. Note that, $R_k=R_{C_k}+R_{S_k}$ for $k=1,2$.

Every node in the network knows these codebooks.

\textbf{Encoding:}

Consider again user $k\in\{1,2\}$.
To send $w_k\in\Wc_k$, user $k$ gets corresponding indices $w_{C_k}$
and $w_{S_k}$. Then user $k$ obtains the following codewords:
\begin{itemize}
  \item From the codebook for $C_k$, user $k$ randomly chooses a
  codeword in bin $w_{C_k}$ according to the uniform distribution,
  where the codeword index is denoted by
  $w_{C_k}^x$ and it gets the corresponding entry of the codebook,
  i.e. $\cv_k^n(w_{C_k},w_{C_k}^x)$.
  \item Similarly, from the codebook for $S_k$, user $k$ randomly chooses a
  codeword in bin $w_{S_k}$ according to the uniform distribution,
  where the codeword index is denoted by
  $w_{S_k}^x$ and it gets the corresponding entry of the codebook,
  i.e. $\sv_k^n(w_{S_k},w_{S_k}^x)$.
  \item Finally, from the codebook for $O_k$, it randomly chooses
  a codeword, which is denoted by $\ov_k^n(w_{O_k}^x)$.
\end{itemize}
Then, user $k$, generates the channel inputs $\xv_k^n$, where each
entry is chosen according to $p(x_k|c_k,s_k,o_k)$ using the codewords
$\cv_k^n(w_{C_k},w_{C_k}^x)$, $\sv_k^n(w_{S_k},w_{S_k}^x)$, and
$\ov_k^n(w_{O_k}^x)$; and it transmits the constructed $\xv_k^n$
in $n$ channel uses. See Fig.~\ref{fig:3}.

\begin{figure}[t] 
    \centering
    \includegraphics[width=0.5\columnwidth]{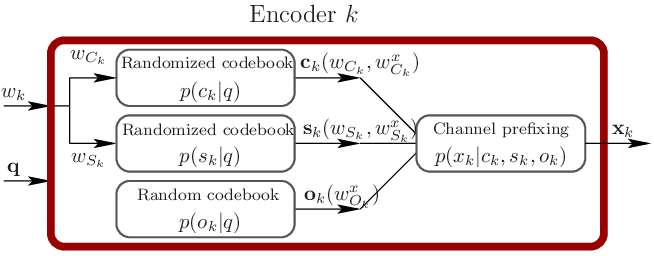}
    \caption{Proposed encoder structure for the IC-E.}\label{fig:3}
\end{figure}

\textbf{Decoding:}

Here we remark that although each user needs to decode only its own message,
we also require receivers to decode common and other information of
the other transmitter.
Suppose receiver $1$ has received $\yv_1^n$.
Let $A_{1,\epsilon}^n$ denote the set of typical $(\qv^n,\cv_1^n,
\sv_1^n,\cv_2^n,\ov_2^n,\yv_1^n)$ sequences.
Decoder $1$ chooses $(w_{C_1},w_{C_1}^x,w_{S_1},w_{S_1}^x,
w_{C_2},w_{C_2}^x,w_{O_2}^x)$ s.t.
$(\qv^n,\cv_1^n(w_{C_1},w_{C_1}^x),\sv_1^n(w_{S_1},w_{S_1}^x),
\cv_2^n(w_{C_2},w_{C_2}^x),\ov_2^n(w_{O_2}^x),$
$\yv_1^n)\in A_{1,\epsilon}^n,$
if such a tuple exists and is unique. Otherwise, the decoder declares an error.
Decoding at receiver $2$ is symmetric and a description of it can
be obtained by reversing the indices $1$ and $2$ above.

Refer to Appendix~\ref{sec:ProofInnerBound} for the rest of the proof.
\end{proof}

The following remarks are now in order.

\begin{enumerate}
  \item The auxiliary random variable $Q$ serves as a time-sharing parameter.
  \item The auxiliary variable $C_1$ is used to construct the \emph{common} secure
  signal of transmitter $1$ that has to be decoded at both receivers, where
  the randomized encoding technique of~\cite{Wyner:The75} is used for this
  construction. Similarly, $C_2$ is used for the \emph{common} secured signal of user $2$.
  \item The auxiliary variable $S_1$ is used to construct the \emph{self} secure signal
  that has to be decoded at receiver $1$ but not at receiver $2$, where
  the randomized encoding technique of~\cite{Wyner:The75} is used for
  this construction. Similarly, $S_2$ is used for the \emph{self} secure signal of user $2$.
  \item The auxiliary variable $O_1$ is used to construct the \emph{other} signal of
  transmitter $1$ that has to be decoded at receiver $2$ but not at
  receiver $1$ (conventional random coding~\cite{ThomasAndCover} is used
  to construct this signal). Similarly, $O_2$ is used for the \emph{other}
  signal of user $2$. Note that we use $R_{O_1}^x$, $R_{O_2}^x$, and set $R_{O_1}=R_{O_2}=0$.
  \item Compared to the Han-Kobayashi scheme~\cite{Han:A81}, the common and
self random variables are constructed with randomized (two-dimensional) codebooks.
This way, they are used not only for transmitting information, but
also for adding randomness.
Moreover, we have two additional random variables in this achievability
scheme. These extra random
variables, namely $O_1$ and $O_2$, are used to facilitate
cooperation among the network users by adding extra randomness
to the channel which has to be decoded by the non-intended receiver. We
note that, compared to random variables $C_k$ and $S_k$, the randomization
added via $O_k$ is considered as interference at the receiver $k$.

\item The gain that can be leveraged from cooperative encoding
can be attributed to the freedom in the allocation of randomization
rates at the two users (e.g., see~\eqref{eq:R_e(p)}). This allows the
users to cooperatively add randomness to impair the eavesdropper with a
minimal impact on the achievable rate at the legitimate receivers. Cooperative
channel prefixing, on the other hand, can be achieved by the joint
optimization of the probabilistic channel prefixes
$p(x_1|c_1,s_1,o_1)$ and $p(x_2|c_2,s_2,o_2)$.
\end{enumerate}


\subsection{Outer Bounds}

\begin{theorem}\label{thm:RO1}
For any $(R_1,R_2) \in \mathbb{C}^{\textrm{IC-E}}$,
\begin{eqnarray}
R_1 & \leq & \max\limits_{p\in\Pc_{O}}\:
I(V_1;Y_1|V_2,U)-I(V_1;Y_e|U)\\
R_2 & \leq & \max\limits_{p\in\Pc_{O}}\:
I(V_2;Y_2|V_1,U)-I(V_2;Y_e|U),
\end{eqnarray}
where
$\Pc_{O}$ is the set of joint distributions that factors as
$$p(u,v_1,v_2,x_1,x_2)
= p(u)p(v_1|u)p(v_2|u)p(x_1|v_1)p(x_2|v_2).$$
\end{theorem}
\begin{proof}
Refer to Appendix~\ref{sec:ProofOuterBound1}.
\end{proof}

\begin{theorem}\label{thm:RO2}
For channels satisfying
\begin{eqnarray}\label{eq:ProofOutCond1}
I(V_2;Y_2|V_1)\leq I(V_2;Y_1|V_1)
\end{eqnarray}
for any distribution that factors as
$p(v_1,v_2,x_1,x_2)=p(v_1)p(v_2)p(x_1|v_1)p(x_2|v_2)$,
an upper bound on the sum-rate of the IC-E is given by

\begin{eqnarray}
R_1 + R_2 & \leq & \max\limits_{p\in\Pc_{O}}\:
I(V_1,V_2;Y_1|U)-I(V_1,V_2;Y_e|U),
\end{eqnarray}
where $U$, $V_1$, and $V_2$ are auxiliary random variables, and
$\Pc_{O}$ is the set of joint distributions that factors as
$$p(u,v_1,v_2,x_1,x_2)
= p(u)p(v_1|u)p(v_2|u)p(x_1|v_1)p(x_2|v_2).$$
\end{theorem}
\begin{proof}
Refer to Appendix~\ref{sec:ProofOuterBound2}.
\end{proof}

The previous sum-rate upper bound also holds for the set of
channels satisfying
\begin{eqnarray}
I(V_2;Y_2)\leq I(V_2;Y_1)
\end{eqnarray}
for any distribution that factors as
$p(v_1,v_2,x_1,x_2)=p(v_1)p(v_2)p(x_1|v_1)p(x_2|v_2)$.
Finally, it is evident that one can obtain another upper bound
by reversing the indices $1$ and $2$ in above expressions.


\subsection{Special Cases}

This section focuses on few special cases, where sharp results
on the secrecy capacity region can be derived. In all these
scenarios, achievability is established using the proposed
cooperative encoding scheme. To simplify the
presentation, we first define the following set of
probability distributions. For random variables $T_1$ and $T_2$,
\begin{eqnarray}
\Pc(T_1,T_2) &\triangleq& \big\{p(q,t_1,t_2,x_1,x_2)\: | \:
p(q,t_1,t_2,x_1,x_2) \:{=}
p(q)p(t_1|q)p(t_2|q)p(x_1|t_1)p(x_2|t_2)\big\}.\nonumber
\end{eqnarray}

\begin{corollary}\label{thm:C1}
If the IC-E satisfies
\begin{eqnarray}\label{eq:Speceq1}
I(V_2;Y_2|V_1,Q) \leq I(V_2;Y_e|Q) \nonumber\\
I(V_2;Y_e|V_1,Q) \leq I(V_2;Y_1|Q)
\end{eqnarray}
for all input distributions that factors as
$p(q)p(v_1|q)p(v_2|q)p(x_1|v_1)p(x_2|v_2)$,
then its secrecy capacity region is given by
$$\mathbb{C}^{\textrm{IC-E}} = \textrm{ the closure of }
\left\{\bigcup\limits_{p\in\Pc(S_1,O_2)} \Rc_{S1}(p)\right\},$$
where
$\Rc_{S1}(p)$ is the set of rate-tuples $(R_1,R_2)$ satisfying
\begin{eqnarray}
R_1 & \leq & \left [I(S_1;Y_1|O_2,Q)-I(S_1;Y_e|Q)\right]^+\nonumber\\
R_2 & = & 0,
\end{eqnarray}
for any $p\in\Pc(S_1,O_2)$.
\end{corollary}

\begin{proof}
Refer to Appendix~\ref{sec:ProofC1}.
\end{proof}

\begin{corollary}\label{thm:C2}
If the IC-E satisfies
\begin{eqnarray}\label{eq:Speceq2}
I(V_2;Y_e|Q)  \leq  I(V_2;Y_1|Q)  \leq  I(V_2;Y_2|Q) \nonumber\\
I(V_1;Y_e|V_2,Q)  \leq  I(V_1;Y_1|V_2,Q) \nonumber\\
I(V_2;Y_2|V_1,Q)  \leq  I(V_2;Y_1|V_1,Q)
\end{eqnarray}
for all input distributions that factors as
$p(q)p(v_1|q)p(v_2|q)p(x_1|v_1)p(x_2|v_2)$,
then its secrecy sum capacity is given as follows.
\begin{eqnarray}
\max\limits_{(R_1,R_2)\in\Cc^{\textrm{IC-E}}} R_1+R_2 
=\max\limits_{p\in\Pc(S_1,C_2)}
I(S_1,C_2;Y_1|Q)-I(S_1,C_2;Y_e|Q).\nonumber
\end{eqnarray}
\end{corollary}

\begin{proof}
Refer to Appendix~\ref{sec:ProofC2}.
\end{proof}

\begin{corollary}\label{thm:C3}
If the IC-E satisfies
\begin{eqnarray}\label{eq:Speceq3}
I(V_2;Y_e|Q)  \leq  I(V_2;Y_1|V_1,Q)  \leq  I(V_2;Y_e|V_1,Q) \nonumber\\
I(V_2;Y_2|V_1,Q)  \leq  I(V_2;Y_1|V_1,Q)
\end{eqnarray}
for all input distributions that factors as
$p(q)p(v_1|q)p(v_2|q)p(x_1|v_1)p(x_2|v_2)$,
then its secrecy sum capacity is given as follows.
\begin{eqnarray}
\max\limits_{(R_1,R_2)\in\Cc^{\textrm{IC-E}}} R_1+R_2 
=\max\limits_{p\in\Pc(S_1,O_2)}
I(S_1,O_2;Y_1|Q)-I(S_1,O_2;Y_e|Q).\nonumber
\end{eqnarray}
We also note that, in this case, $O_2$ will not increase
the sum-rate, and hence, we can set $|\Oc_2|=1$
\end{corollary}

\begin{proof}
Refer to Appendix~\ref{sec:ProofC3}.
\end{proof}

Another case for which the cooperative encoding
approach can attain the sum-capacity is the following.
\begin{corollary}\label{thm:C4}
If the IC-E satisfies
\begin{eqnarray}\label{eq:Speceq4}
I(V_2;Y_1|Q)  \leq  I(V_2;Y_e|V_1,Q)  \leq  I(V_2;Y_1|V_1,Q) \nonumber\\
I(V_2;Y_2|V_1,Q)  \leq  I(V_2;Y_1|V_1,Q)
\end{eqnarray}
for all input distributions that factors as
$p(q)p(v_1|q)p(v_2|q)p(x_1|v_1)p(x_2|v_2)$,
then its secrecy sum capacity is given as follows.
\begin{eqnarray}
\max\limits_{(R_1,R_2)\in\Cc^{\textrm{IC-E}}} R_1+R_2 
=\max\limits_{p\in\Pc(S_1,O_2)} 
I(S_1,O_2;Y_1|Q)-I(S_1,O_2;Y_e|Q).\nonumber
\end{eqnarray}
\end{corollary}

\begin{proof}
Refer to Appendix~\ref{sec:ProofC4}.
\end{proof}

Now, we use our results on the IC-E to shed more light on
the secrecy capacity of the discrete memoryless multiple
access channel. In particular, it is easy to see that the
multiple access channel with an eavesdropper (MAC-E) defined
by $p(y_1,y_e|x_1,x_2)$ is equivalent to the IC-E defined
by $p(y_1,y_2,y_e|x_1,x_2)=p(y_1,y_e|x_1,x_2) \delta(y_2-y_1)$.
This allows for specializing the results obtained earlier to the MAC-E.

\begin{corollary}
$$\Rc^{\textrm{MAC-E}}\triangleq \textrm{ the closure of}
\left\{\bigcup\limits_{p\in\Pc(C_1,C_2)} \Rc(p)\right\},
$$
where the channel is given by
$p(y_1,y_e|x_1,x_2)\delta(y_2-y_1)$.
\end{corollary}

Furthermore, the following result characterizes the
secrecy sum rate of the weak MAC-E.

\begin{corollary}[MAC-E with a weak eavesdropper]\label{thm:C6}
If the eavesdropper is weak for the MAC-E, i.e.,
\begin{eqnarray}\label{eq:Speceq7}
I(V_1;Y_e|V_2) \leq I(V_1;Y_1|V_2)\nonumber\\
I(V_2;Y_e|V_1) \leq I(V_2;Y_1|V_1),
\end{eqnarray}
for all input distributions that factor as
$p(v_1)p(v_2)p(x_1|v_1)p(x_2|v_2)$, then the
secure sum-rate capacity is characterized as the following.
\begin{eqnarray}
\max\limits_{(R_1,R_2)\in\Cc^{\textrm{MAC-E}}} R_1+R_2 
= \max\limits_{p\in\Pc(C_1,C_2)} 
I(C_1,C_2;Y_1|Q) -I(C_1,C_2;Y_e|Q)\nonumber
\end{eqnarray}
\end{corollary}

\begin{proof}
Refer to Appendix~\ref{sec:ProofC6}.
\end{proof}

Another special case of our model is the relay-eavesdropper
channel with a deaf helper. In this scenario, transmitter $1$ has
a secret message for receiver $1$ and transmitter $2$
is only interested in helping transmitter
$1$ in increasing its secure transmission rates. Here,
the random variable $O_2$ at transmitter $2$ is
utilized to add randomness to the network. Again, the
regions given earlier
can be specialized to this scenario. For example,
the following region is achievable
for this relay-eavesdropper model.

\begin{eqnarray}
\Rc^{\textrm{RE}}&\triangleq&\textrm{ the closure of the convex hull of }
\left\{\bigcup\limits_{p\in\Pc(S_1,O_2,|\Qc|=1)} \Rc(p)\right\},
\end{eqnarray}
where $\Pc(S_1,O_2,|\Qc|=1)$ denotes the probability distributions
in $\Pc(S_1,O_2)$ with a deterministic $Q$.

For this relay-eavesdropper scenario, the noise forwarding (NF) scheme
proposed in~\cite{Lai:The08} achieves the following rate.
\begin{equation}\label{eq:R_NF}
R^{\textrm{[NF]}}=\max\limits_{p\in\Pc(S_1,O_2,|\Qc|=1)}
\: R_1(p) ,
\end{equation}
where

$R_1(p)\triangleq [I(S_1;Y_1|O_2)+\min\{I(O_2;Y_1),I(O_2;Y_e|S_1)\}$
$-\min\{I(O_2;Y_1),I(O_2;Y_e)\}-I(S_1;Y_e|O_2)]^+$.

The following result shows that NF is a special
case of the cooperative encoding scheme and provides a simplification
of the achievable secrecy rate.

\begin{corollary}\label{thm:C7}
$(R^{\textrm{[NF]}},0)\in\Rc^{\textrm{RE}}$, where
$R^{\textrm{[NF]}}$ can be simplified as follows.
\begin{eqnarray}
R^{\textrm{[NF]}}=
\max\limits_{p\in\Pc(S_1,O_2,|\Qc|=1) \textrm{ s.t. }
I(O_2;Y_e)\leq I(O_2;Y_1)} I(S_1;Y_1|O_2) 
+\min\{I(O_2;Y_1),I(O_2;Y_e|S_1)\}-I(S_1,O_2;Y_e).
\end{eqnarray}
\end{corollary}

\begin{proof}
Refer to Appendix~\ref{sec:ProofC7}.
\end{proof}

Finally, the next result establishes the optimality
of NF in certain relay-eavesdropper channels.

\begin{corollary}\label{thm:C8}
Noise Forwarding scheme is optimal for the relay-eavesdropper
channels which satisfy
\begin{eqnarray}\label{eq:Speceq8}
I(V_2;Y_1) \leq I(V_2;Y_e|V_1),
\end{eqnarray}
for all input distributions that factor as
$p(v_1)p(v_2)p(x_1|v_1)p(x_2|v_2)$,
and the corresponding secrecy capacity is
\begin{eqnarray}
\Cc^{\textrm{RE}} =
\max\limits_{
\footnotesize
\begin{array}{c}
  p\in\Pc(S_1,O_2,|\Qc|=1) \\
  \textrm{s.t. } I(O_2;Y_e)\leq I(O_2;Y_1)
\end{array}}
\normalsize
I(S_1,O_2;Y_1)-I(S_1,O_2;Y_e).\nonumber
\end{eqnarray}
\end{corollary}

\begin{proof}
Refer to Appendix~\ref{sec:ProofC8}.
\end{proof}


\section{Security for the Gaussian Interference Channel}
\label{sec:III}
\subsection{Inner Bound and Numerical Results}

In its standard form~\cite{Carleial:Interference78}, the two user
Gaussian Interference Channel with an Eavesdropper (GIC-E) is given by
\begin{eqnarray}\label{eq:GIC-E}
Y_1 & = & X_1 + \sqrt{c_{21}} X_2 + N_1 \nonumber\\
Y_2 & = & \sqrt{c_{12}} X_1 + X_2 + N_2 \\
Y_e & = & \sqrt{c_{1e}} X_1 + \sqrt{c_{2e}} X_2 + N_e, \nonumber
\end{eqnarray}
where $N_{r}\sim\Nc(0,1)$ is the noise at each receiver $r=1,2,e$
and the average power constraints
are $\frac{1}{n}\sum\limits_{t=1}^n
(X_k(t))^2\leq P_k$ for $k=1,2$. The secrecy capacity
region of the GIC-E is denoted as
$\mathbb{C}^{\textrm{GIC-E}}$.


The goal here is to specialize the results obtained in the
previous section to the Gaussian scenario and illustrate the
gains that can be leveraged from the cooperative coding for randomized
codebooks and for channel prefixing, and from time sharing. For this
scenario, the Gaussian codebooks
are used and the same regions will be achievable after taking into account
the power constraint at the users.
Towards this end, we will need the following definitions.
Consider a probability mass
function on the time sharing parameter denoted by $p(q)$.
Let $\Ac(p(q))$ denote the set of all
possible power allocations, i.e.,

$\Ac(p(q)) \triangleq \bigg\{\big(P_1^c(q),
P_1^s(q),P_1^o(q),P_1^j(q),P_2^c(q),P_2^s(q),$
$P_2^o(q),P_2^j(q)\big)\big|
\sum\limits_{q\in\Qc}
(P_k^c(q)+P_k^s(q)+P_k^o(q)+P_k^j(q))p(q) \leq P_k,
\textrm{ for } k=1,2. \bigg\}$

Now, we define a set of joint distributions
$\Pc_G$ for the Gaussian case as follows.

\begin{eqnarray}
\Pc_G \triangleq  \bigg\{p &|& p\in\Pc,
(P_1^c(q),P_1^s(q),P_1^o(q),P_1^j(q),P_2^c(q),P_2^s(q),P_2^o(q),P_2^j(q))\in\Ac(p(q)), \nonumber\\
&&C_1(q)\sim\Nc(0,P_1^c(q)), S_1(q)\sim\Nc(0,P_1^s(q)), O_1(q)\sim\Nc(0,P_1^o(q)), J_1(q)\sim\Nc(0,P_1^j(q)), \nonumber\\
&&C_2(q)\sim\Nc(0,P_2^c(q)), S_2(q)\sim\Nc(0,P_2^s(q)), O_2(q)\sim\Nc(0,P_2^o(q)), J_2(q)\sim\Nc(0,P_2^j(q)), \nonumber\\
&&X_1(q)=C_1(q)+S_1(q)+O_1(q)+J_1(q), X_2(q)=C_2(q)+S_2(q)+O_2(q)+J_2(q) \bigg\}, \nonumber
\end{eqnarray}
where the Gaussian model given in (\ref{eq:GIC-E})
gives $p(y_1,y_2,y_e|x_1,x_2)$. Using this set of distributions,
we obtain the following achievable secrecy rate region for the GIC-E.

\begin{corollary}\label{thm:RG}

$\Rc^{\textrm{GIC-E}}
\triangleq
\textrm{ the closure of }
\left\{
\bigcup\limits_{p\in \Pc_G} \Rc(p)
\right\}
\subset  \mathbb{C}^{\textrm{GIC-E}}$.

\end{corollary}

It is interesting to see that our particular choice of the
channel prefixing distribution $p(x_k|c_k,s_k,o_k)$ in the
above corollary corresponds to a superposition coding
approach where $X_k=C_k+S_k+O_k+J_k$. This observation
establishes the fact that noise injection scheme
of~\cite{Negi:Secret05} and jamming scheme of~\cite{Tekin:The08}
are {\bf special cases} of the channel prefixing
technique of~\cite{Csiszar:Broadcast78}.

The following computationally simple subregion
will be used to generate some of our numerical results.

\begin{corollary}\label{thm:RG3}
$\Rc_2^{\textrm{GIC-E}}
\subset \Rc^{\textrm{GIC-E}} \subset  \mathbb{C}^{\textrm{GIC-E}}$,
where
$\Rc_2^{\textrm{GIC-E}}
\triangleq
\textrm{ the convex closure of }
\big\{
\bigcup\limits_{p\in\Pc_{G2}} \Rc(p)
\big\},$
and
$\Pc_{G2}\triangleq \{p \: | \: p\in\Pc_G,
|\Qc|=1,P_1^s(q)=P_1^o(q)=P_2^s(q)=P_2^o(q)=0
\textrm{ for any } Q=q \}.$
\end{corollary}

Another simplification can be obtained from the following
TDMA-like approach. Here we divide the $n$ channel uses
into two parts of lengths represented by $\alpha n$ and
$(1-\alpha)n$, where $0\leq \alpha \leq 1$ and $\alpha n$
is assumed to be an integer. During the first period, transmitter
$1$ generates randomized codewords using power $P_1^s(1)$ and
transmitter $2$ jams the channel using power $P_2^j(1)$. For the
second period, the roles of the users are reversed, where the
users use powers $P_2^s(2)$ and $P_1^j(2)$. We refer to
this scheme cooperative TDMA (C-TDMA) which achieves
the following region.

\begin{corollary}\label{thm:RC-TDMA}
$\Rc_{\textrm{C-TDMA}} \subset \Rc^{\textrm{GIC-E}} \subset
\mathbb{C}^{\textrm{GIC-E}}$, where
$$\Rc_{\textrm{C-TDMA}}\triangleq \textrm{ the closure of }
\left\{
\bigcup\limits_{{}^{\quad\quad\quad \alpha  \in  [0,1]}_{
{}^{\alpha P_1^s(1) + (1-\alpha)P_1^j(2) \leq P_1}_
{\alpha P_2^j(1) + (1-\alpha)P_2^s(2) \leq P_2}
}}  (R_1,R_2)
\right\},$$ where

\begin{eqnarray}
R_1=\frac{\alpha}{2} \Bigg[
\log\left( 1+\frac{P_1^s(1)}{1+c_{21}P_2^{j}(1)} \right)-
\quad \log\left( 1+\frac{c_{1e}P_1^s(1)}{1+c_{2e}P_2^{j}(1)} \right)
\Bigg]^+,
\end{eqnarray}
and
\begin{eqnarray}
R_2=\frac{(1-\alpha)}{2} \Bigg[
\log\left( 1+\frac{P_2^s(2)}{1+c_{12}P_1^{j}(2)} \right)-
\log\left( 1+\frac{c_{2e}P_2^s(2)}{1+c_{1e}P_1^{j}(2)} \right)
\Bigg]^+.
\end{eqnarray}
\end{corollary}

\begin{proof}
This is a subregion of the $\mathbb{C}^{\textrm{GIC-E}}$, where
we use a time sharing random variable satisfying
$p(q=1)=\alpha$ and $p(q=2)=1-\alpha$, and utilize
the random variables $S_1$ and $S_2$.
The proof also follows by respective single-user
Gaussian wiretap channel result~\cite{Leung-Yan-Cheong:The78}
with the modified noise variances due to the jamming signals.
\end{proof}

In the C-TDMA scheme above, we only add randomness by noise injection
at the {\em helper} node. However, our cooperative encoding
scheme (Corollary~\ref{thm:RG}) allows
for the implementation of more general {\em cooperation}
strategies. For example, in a more general TDMA approach, each user can help
the other via both the design of randomized codebook and channel
prefixing (i.e., the noise forwarding scheme described in
Section~\ref{sec:GaussianRelay-EavesdropperChannel}).
In addition, one can develop enhanced transmission strategies with
a time-sharing random variable of cardinality greater than $2$.

We now provide numerical results for the following
subregions of the achievable region given in Corollary~\ref{thm:RG}.

\begin{itemize}
  \item $\Rc_2^{\textrm{GIC-E}}$: Here
  we utilize both cooperative randomized codebook design and channel prefixing.
  \item $\Rc_2^{\textrm{GIC-E}}{\textrm{(rc or cp)}}$: Here we utilize
  either cooperative randomized codebook design (rc) or channel
  prefixing (cp) scheme at a transmitter, but not both.
  \item $\Rc_2^{\textrm{GIC-E}}{\textrm{(ncp)}}$: Here we only utilize
  cooperative randomized codebook design, no channel prefixing
  (ncp) is implemented.
  \item $\Rc_{\textrm{C-TDMA}}$: This region is an example of utilizing both
  time-sharing and cooperative channel prefixing.
  \item $\Rc_{\textrm{C-TDMA}}{\textrm{(ncp)}}$: This region is a subregion of $\Rc_{\textrm{C-TDMA}}$,
  for which we set the jamming powers to zero.
\end{itemize}

\begin{figure}[t] 
    \centering
    \includegraphics[width=0.5\columnwidth]{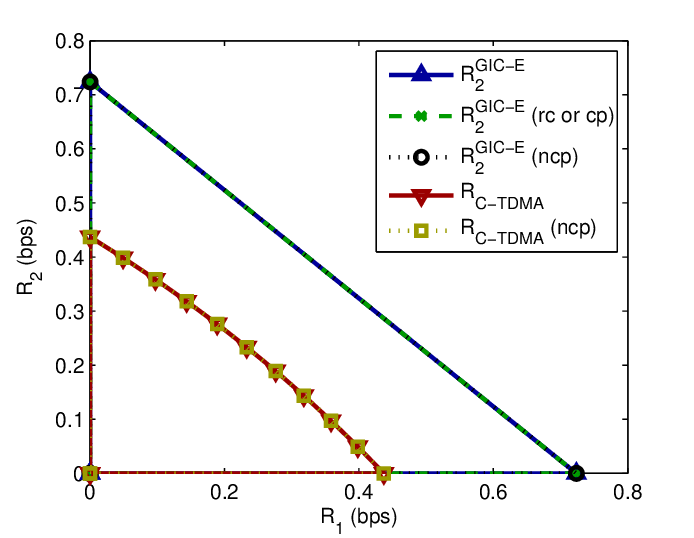}
    \caption{Numerical results for GIC-E with $c_{12}=c_{21}=1.9$,
    $c_{1e}=c_{2e}=0.5$, $P_{1}=P_{2}=10$.
    The three schemes, performance of which are given by
    $R_2^{\textrm{GIC-E}}$, $R_2^{\textrm{GIC-E}} \textrm{(rc or cp)}$, and
    $R_2^{\textrm{GIC-E}}\textrm{(ncp)}$, have the same performance and
    outperform the ones represented by
    $R_{\textrm{C-TDMA}}$ and $R_{\textrm{C-TDMA}}\textrm{(ncp)}$, which
    achieve the same region.
    }\label{fig:4}
\end{figure}

\begin{figure}[t] 
    \centering
    \includegraphics[width=0.5\columnwidth]{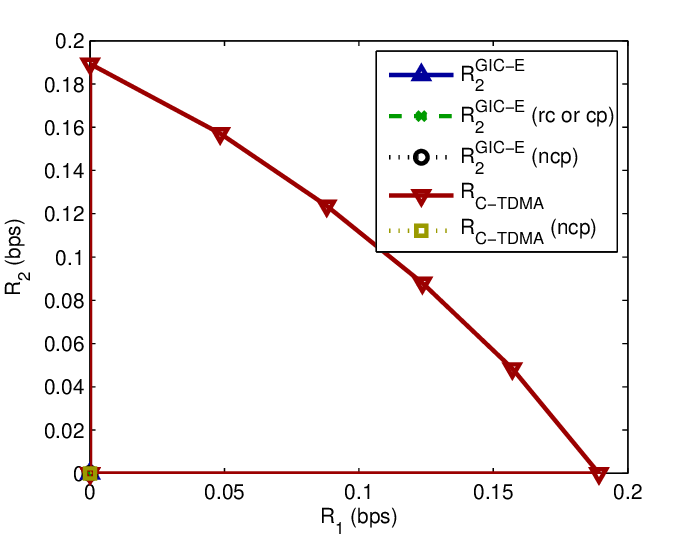}
    \caption{Numerical results for GIC-E with $c_{12}=c_{21}=0.6$,
    $c_{1e}=c_{2e}=1.1$, $P_{1}=P_{2}=10$.
    Only the scheme represented by $R_{\textrm{C-TDMA}}$ achieves positive rates.
    }\label{fig:5}
\end{figure}

\begin{figure}[t] 
    \centering
    \includegraphics[width=0.5\columnwidth]{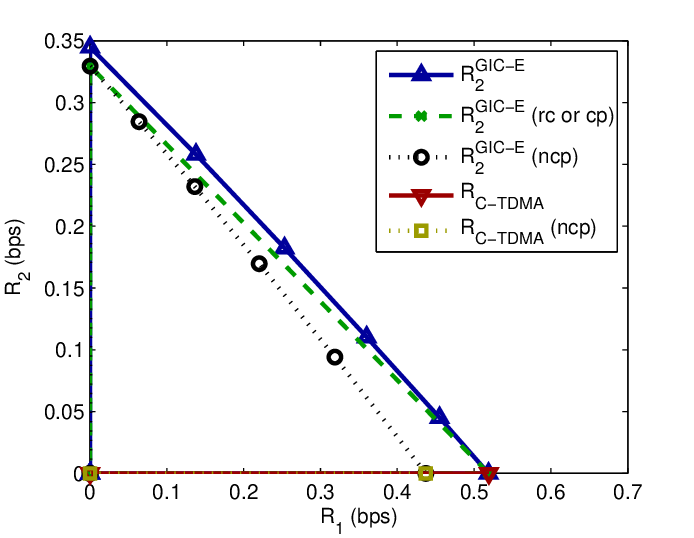}
    \caption{Numerical results for GIC-E with $c_{12}=1.9$,
    $c_{21}=1$, $c_{1e}=0.5$, $c_{2e}=1.6$, $P_{1}=P_{2}=10$.
    The schemes represented by $R_{\textrm{C-TDMA}}$ and
    $R_{\textrm{C-TDMA}}\textrm{(ncp)}$ does not achieve positive
    rates for user $2$.
    }\label{fig:6}
\end{figure}

The first
scenario depicted in Fig.~\ref{fig:4} shows the gain offered
by the cooperative encoding technique, as compared with the various
cooperative TDMA approaches. Also, it is shown that cooperative
channel prefixing does not increase the secrecy rate region in
this particular scenario. In Fig.~\ref{fig:5}, we consider a
channel with a rather capable eavesdropper.
In this case, it is straightforward to verify that the
corresponding single user channels have zero
secrecy capacities. However, with the appropriate
cooperation strategies between the two {\bf interfering
users}, the two users can achieve non-zero rates (as
reported in the figure). In Fig.~\ref{fig:6}, we consider an asymmetric
scenario, in which the first user has a weak channel to the eavesdropper,
but the second user has a strong channel to the eavesdropper.
In this case, the proposed cooperative encoding technique allows the
second user to achieve a positive secure transmission rate, which is not
possible by exploiting only the channel prefixing and time-sharing techniques.
In addition, by prefixing the channel,
the second user can help the first one to increase
its secure transmission rate. Finally,
we note that for some channel coefficients
$\Rc_{\textrm{C-TDMA}}$ outperforms $\Rc_2^{\textrm{GIC-E}}$ and
for some others $\Rc_2^{\textrm{GIC-E}}$ outperforms $\Rc_{\textrm{C-TDMA}}$.
Therefore, in general, the proposed techniques (cooperative randomized codebook
design, cooperative channel prefixing, and time-sharing) should be
exploited simultaneously as considered in Corollary~\ref{thm:RG}.


\subsection{Special Cases}

\subsubsection{The Multiple Access Channel}
First, we define a set of probability distributions
$$
\Pc_{G3} \triangleq  \bigg\{p \: | \: p\in\Pc_G,
P_1^s(q)=P_1^o(q)=P_2^s(q)=P_2^o(q)=0
\textrm{ for any } Q=q \bigg\}.$$

Using this notation, one can easily see that the region
$\Rc^{\textrm{GIC-E}}$ in Corollary~\ref{thm:RG}
reduces to the following achievable secrecy rate
region for the Gaussian Multiple
Access Channel with an Eavesdropper (GMAC-E).
$$\Rc^{\textrm{GMAC-E}}\triangleq
\textrm{ the closure of }
\left\{\bigcup\limits_{p\in\Pc_{G3}} \Rc(p)\right\},$$
where the expressions in the region $\Rc(p)$
are calculated for the channel given by $p(y_1,y_e|x_1,x_2)\delta(y_2-y_1)$.

The region $\Rc^{\textrm{GMAC-E}}$ generalizes the
one obtained in~\cite{Tekin:The08} for the two user case.
The underlying reason is that, in the
achievable scheme of~\cite{Tekin:The08}, the users
are either transmitting their codewords or jamming
the channel whereas, in our approach, the users can transmit their
codewords and jam the channel simultaneously. In addition,
our cooperative TDMA approach generalizes the one proposed
in~\cite{Tekin:The08}, as we allow the two users to cooperate
in the design of randomized codebooks and channel prefixing
during the time slots dedicated to either one.

\subsubsection{The Relay-Eavesdropper Channel}
\label{sec:GaussianRelay-EavesdropperChannel}

In the previous section, we argued that the noise forwarding (NF)
scheme of~\cite{Lai:The08} can be obtained as a special case
of our generalized cooperation scheme. Here, we demonstrate the positive
impact of channel prefixing on increasing the achievable
secrecy rate of  the Gaussian relay-eavesdropper channel.
In particular, the proposed region for the Gaussian IC-E, when
specialized to the Gaussian relay-eavesdropper setting, results in
$$\Rc^{\textrm{GRE}}\triangleq
\textrm{ closure of the convex hull of }
\left\{\bigcup\limits_{p\in\Pc_{G4}} \Rc(p)\right\},$$
where
$$
\Pc_{G4}  \triangleq \bigg\{p \: | \: p\in\Pc_G,
|\Qc|=1, P_1^c(q)=P_1^o(q)=P_2^c(q)=
P_2^s(q)=0\bigg\}.$$

One the other hand, noise forwarding with
no channel prefixing (GNF-ncp) results in the following achievable rate.
\begin{eqnarray}\label{eq:RGNF}
R^{\textrm{[GNF-ncp]}}&=&\bigg[
\frac{1}{2}\log\left(1 + P_1 \right)
-
\frac{1}{2}\log\left(1 + c_{1e}P_1 \right)\nonumber\\
-&&
\min\left\{
    \frac{1}{2}\log\left(1 + \frac{c_{21}P_2}{1+P_1} \right),
    \frac{1}{2}\log\left(1 + \frac{c_{2e}P_2}{1+c_{1e}P_1} \right)
\right\}\nonumber\\
+&&
\min\left\{
    \frac{1}{2}\log\left(1 + \frac{c_{21}P_2}{1+P_1} \right),
    \frac{1}{2}\log\left(1 + c_{2e}P_2 \right)
\right\}
\bigg]^+,
\end{eqnarray}
where we choose $X_1=S_1\sim\Nc(0,P_1)$ and
$X_2=O_2\sim\Nc(0,P_2)$ in the expression of $R^{\textrm{[NF]}}$
(see also~\cite{Lai:The08}).

Numerically, the positive impact of channel prefixing is
illustrated in the following example. First, it is easy
to see that the following secrecy rate is achievable with channel prefixing
\begin{equation}
R_{1}=\left[ \frac{1}{2}\log\left(1+\frac{P_1}{1+c_{21}P_2}\right) -
\frac{1}{2}\log\left(1+\frac{c_{1e}P_1}{1+c_{2e}P_2}\right) \right]^+,\nonumber
\end{equation}
since $(R_{1},0)\in\Rc^{\textrm{GRE}}$ (i.e., we set
$P_1^s=P_1$ and $P_2^j=P_2$). Now, we let $c_{1e}=c_{2e}=1$ and
$P_1=P_2=1$, resulting in
$R^{\textrm{[GNF-ncp]}}=0$ and $R_{1}>0$ if $c_{21}<1$.


\section{Conclusions}
\label{sec:IV}

This work considered the two-user interference channel with an
(external) eavesdropper. An inner bound on the achievable secrecy
rate region was derived using a scheme that combines cooperative
randomized codebook design, channel prefixing, message
splitting, and time-sharing techniques.
More specifically, our achievable scheme allows the two users
to cooperatively construct their randomized codebooks and channel
prefixing distributions. Outer bounds are then derived and
used to establish the optimality of the proposed scheme in some
special cases. For the Gaussian scenario, channel prefixing
was used to allow the users to transmit independently generated
noise samples using a fraction of the available power.
Moreover, as a special case of time sharing, we have developed a novel
cooperative TDMA scheme, where a user can add {\em structured and
unstructured} noise to the channel during the allocated slot for the
other user. It is shown that this scheme reduces to the noise forwarding
scheme proposed earlier for the relay-eavesdropper
channel. In the Gaussian multiple-access setting, our
cooperative encoding and channel prefixing scheme was shown to enlarge
the achievable regions obtained in previous works.
The most interesting aspect of our results is,
perhaps, the illumination of the role of interference in
cooperatively adding randomness to increase
the achievable secrecy rates in multi-user networks.


\appendices

\section{Proof of Theorem~\ref{thm:R}}
\label{sec:ProofInnerBound}

\textbf{Probability of Error Analysis:}

Below we show that the decoding error probability of user $k$ averaged over
the ensemble can be arbitrarily made small for sufficiently large $n$.
This demonstrates the existence of a codebook with the property that
$\max(P_{e,1},P_{e,2})\leq \epsilon$, for any given $\epsilon>0$.
The analysis follows from similar arguments given in~\cite{Han:A81}.
See also~\cite{ThomasAndCover} for joint typical decoding error computations.
Here, for any given $\epsilon>0$, each receiver can decode
corresponding messages given above with an error probability less than
$\epsilon$ as $n\to\infty$, if the rates satisfy the following equations.
\begin{eqnarray}\label{eq:ProofR1}
R_{\Sc} + R_{\Sc}^x \leq I(\Sc;Y_1|\Sc^c,Q),\:\:
\forall \Sc \subset \{C_1,S_1,C_2,O_2\},
\end{eqnarray}
\begin{eqnarray}\label{eq:ProofR2}
R_{\Sc} + R_{\Sc}^x \leq I(\Sc;Y_2|\Sc^c,Q),\:\:
\forall \Sc \subset \{C_1,O_1,C_2,S_2\}.
\end{eqnarray}

\textbf{Equivocation Computation:}

We first write the following.

\begin{eqnarray}\label{eq:ProofEquivocation}
H(W_1,W_2|\Ym_e)&=& H(W_{C_1},W_{S_1},W_{C_2},W_{S_2}|\Ym_e)\nonumber\\
&\geq& H(W_{C_1},W_{S_1},W_{C_2},W_{S_2}|\Ym_e,\Qm) \nonumber\\
&=& H(W_{C_1},W_{S_1},W_{C_2},W_{S_2},\Ym_e|\Qm) - H(\Ym_e|\Qm)\nonumber\\
&=& H(\Cm_1,\Sm_1,\Om_1,\Cm_2,\Sm_2,\Om_2|\Qm) - H(\Ym_e|\Qm)\nonumber\\
&&{+}\: H(\Wc,\Ym_e|\Cm_1,\Sm_1,\Om_1,\Cm_2,\Sm_2,\Om_2,\Qm)\nonumber\\
&&{-}\:
H(\Cm_1,\Sm_1,\Om_1,\Cm_2,\Sm_2,\Om_2|\Wc,\Ym_e,\Qm)
\nonumber\\
&\geq& H(\Cm_1,\Sm_1,\Om_1,\Cm_2,\Sm_2,\Om_2|\Qm)- H(\Ym_e|\Qm)\nonumber\\
&&{+}\: H(\Ym_e|\Cm_1,\Sm_1,\Om_1,\Cm_2,\Sm_2,\Om_2,\Qm)\nonumber\\
&&{-}\:
H(\Cm_1,\Sm_1,\Om_1,\Cm_2,\Sm_2,\Om_2|\Wc,\Ym_e,\Qm)
\nonumber\\
&=& H(\Cm_1,\Sm_1,\Om_1,\Cm_2,\Sm_2,\Om_2|\Qm)\nonumber\\
&&{-}\:
H(\Cm_1,\Sm_1,\Om_1,\Cm_2,\Sm_2,\Om_2|\Wc,\Ym_e,\Qm)
\nonumber\\
&&{-}\:
I(\Cm_1,\Sm_1,\Om_1,\Cm_2,\Sm_2,\Om_2;\Ym_e|\Qm),
\end{eqnarray}
where we use the set notation
$\Wc \triangleq (W_{C_1},W_{S_1},W_{C_2},W_{S_2})$
to ease the presentation, and the inequalities are due
to the fact that conditioning does not increase entropy,

Here,

$H(\Cm_1,\Sm_1,\Om_1,\Cm_2,\Sm_2,\Om_2|\Qm)$
\begin{eqnarray}\label{eq:ProofEq1}
=n(R_{C_1}+R_{C_1}^x+R_{S_1}
+R_{S_1}^x+R_{O_1}^x+R_{C_2}+R_{C_2}^x+R_{S_2}+R_{S_2}^x+R_{O_2}^x-6\epsilon_1),
\end{eqnarray}
as, given $\Qm=\qv$, the tuple $(\Cm_1,\Sm_1,\Om_1,\Cm_2,\Sm_2,\Om_2)$
has $2^{n(R_{C_1}+R_{C_1}^x+R_{S_1}+R_{S_1}^x+R_{O_1}^x
+R_{C_2}+R_{C_2}^x+R_{S_2}+R_{S_2}^x+R_{O_2}^x-6\epsilon_1)}$
possible values each with equal probability.

Secondly,
\begin{eqnarray}\label{eq:ProofEq2}
I(\Cm_1,\Sm_1,\Om_1,\Cm_2,\Sm_2,\Om_2;\Ym_e|\Qm)
\leq n I(C_1,S_1,O_1,C_2,S_2,O_2;Y_e|Q)+n\epsilon_2,
\end{eqnarray}
where $\epsilon_2\to 0$ as $n\to\infty$.
See, for example, Lemma~$8$ of~\cite{Wyner:The75}.

Lastly, for any $W_{C_1}=w_{C_1}$, $W_{S_1}=w_{S_1}$,
$W_{C_2}=w_{C_2}$, $W_{S_2}=w_{S_2}$, and $\Qm=\qv$, we have
$$
H(\Cm_1,\Sm_1,\Om_1,\Cm_2,\Sm_2,\Om_2|
W_{C_1}=w_{C_1},W_{S_1}=w_{S_1},W_{C_2}=w_{C_2},W_{S_2}=w_{S_2},\Ym_e,\Qm=\qv)
\leq n\epsilon_3,
$$
for some $\epsilon_3\to 0$ as $n\to\infty$.
This is due to the Fano's inequality together with the randomized codebook
construction: Given all the message (bin) indices of two users, eavesdropper can
decode the randomization indices among those bins. Due to joint typicality,
this latter argument holds as long as the rates satisfy the following equations.
\begin{eqnarray}\label{eq:ProofR3}
R_{\Sc}^x  \leq  I(\Sc;Y_e|\Sc^c,Q),\:
\forall \Sc \subset \{C_1,S_1,O_1,C_2,S_2,O_2\}.
\end{eqnarray}
This follows as given bin indices $W_{C_1}$, $W_{S_1}$, $W_{C_2}$,
and $W_{S_2}$, this reduces to MAC probability
of error computation among the codewords of those bins.
See~\cite{ThomasAndCover} for details of computing error probabilities in MAC.
Then, averaging over $W_{C_1}$, $W_{S_1}$,
$W_{C_2}$, $W_{S_2}$, and $\Qm$, we obtain
\begin{eqnarray}\label{eq:ProofEq3}
H(\Cm_1,\Sm_1,\Om_1,\Cm_2,\Sm_2,\Om_2|
W_{C_1},W_{S_1},W_{C_2},W_{S_2},\Ym_e,\Qm)
\leq n\epsilon_3.
\end{eqnarray}

Hence, using (\ref{eq:ProofEq1}), (\ref{eq:ProofEq2}), and
(\ref{eq:ProofEq3}) in (\ref{eq:ProofEquivocation}) we obtain

\begin{equation}
R_1+R_2-\frac{1}{n}H(W_1,W_2|Y_e^n)
\leq  6\epsilon_1+\epsilon_2+\epsilon_3
\triangleq \epsilon
\to 0,
\end{equation}
as $n\to\infty$, where we set
\begin{equation}\label{eq:ProofR4}
R_{\Sc}^x  =  I(\Sc;Y_e|Q), \: \Sc = \{C_1,S_1,O_1,C_2,S_2,O_2\}.
\end{equation}
Combining (\ref{eq:ProofR1}), (\ref{eq:ProofR2}),
(\ref{eq:ProofR3}), and (\ref{eq:ProofR4}) we obtain the result,
i.e., $\Rc(p)$ is achievable for any $p\in\Pc$.


\section{Proof of Theorem~\ref{thm:RO1}}
\label{sec:ProofOuterBound1}

We bound $R_1$ below. The bound on $R_2$ can be obtained by following
similar steps below and reversing the indices $1$ and $2$.
We first state the following definitions.
For any random variable $Y$,
$\tilde{\Ym}^{i+1}\triangleq [Y(i+1) \cdots Y(n)]$, and
\begin{eqnarray}
I_1 & \triangleq & \sum\limits_{i=1}^n
I(\tilde{\Ym}_e^{i+1};Y_1(i)|\Ym_1^{i-1})\\
\hat{I}_1 & \triangleq & \sum\limits_{i=1}^n
I(\Ym_1^{i-1};Y_e(i)|\tilde{\Ym}_e^{i+1})\\
I_2 & \triangleq & \sum\limits_{i=1}^n
I(\tilde{\Ym}_e^{i+1};Y_1(i)|\Ym_1^{i-1},W_1)\\
\hat{I}_2 & \triangleq & \sum\limits_{i=1}^n
I(\Ym_1^{i-1};Y_e(i)|\tilde{\Ym}_e^{i+1},W_1)
\end{eqnarray}

Then, we consider the following bound.
\begin{eqnarray}
R_1-\epsilon & \leq & \frac{1}{n}H(W_1|\Ym_e) \nonumber\\
& = & \frac{1}{n}\left(H(W_1) - I(W_1;\Ym_e)\right) \nonumber\\
& = & \frac{1}{n}\left(H(W_1|\Ym_1) + I(W_1;\Ym_1) - I(W_1;\Ym_e)\right) \nonumber\\
& \leq & \epsilon_1
+ \frac{1}{n}\Bigg(\sum\limits_{i=1}^n I(W_1;Y_1(i)|\Ym_1^{i-1}) 
{-}\: \sum\limits_{i=1}^n I(W_1;Y_e(i)|\tilde{\Ym}_e^{i+1}) \Bigg)\nonumber\\
& = &  \epsilon_1
+ \frac{1}{n}\bigg(\sum\limits_{i=1}^n I(W_1;Y_1(i)|\Ym_1^{i-1}, \tilde{\Ym}_e^{i+1})
+ I_1 - I_2 
{-} \: \sum\limits_{i=1}^n I(W_1;Y_e(i)|\Ym_1^{i-1}, \tilde{\Ym}_e^{i+1})
- \hat{I}_1 + \hat{I}_2 \bigg) \nonumber\\
& = &  \epsilon_1
+ \frac{1}{n}\Bigg(\sum\limits_{i=1}^n I(W_1;Y_1(i)|\Ym_1^{i-1}, \tilde{\Ym}_e^{i+1}) {-}\:
\sum\limits_{i=1}^n I(W_1;Y_e(i)|\Ym_1^{i-1}, \tilde{\Ym}_e^{i+1}) \Bigg),\nonumber
\end{eqnarray}
where the first inequality is due to Lemma~\ref{thm:secrecylemma}
given at the end of this section, the second inequality is due to the Fano's
inequality at the receiver $1$ with some $\epsilon_1\to 0$ as $n\to \infty$,
the last equality is due to observations $I_1=\hat{I}_1$ and $I_2=\hat{I}_2$
(see~\cite[Lemma $7$]{Csiszar:Broadcast78}).

We define $U(i)\triangleq(\Ym_1^{i-1},\tilde{\Ym}_e^{i+1},i)$,
$V_1(i)\triangleq(U(i),W_1)$, and $V_2(i)\triangleq(U(i),W_2)$.
Using standard techniques (see, e.g.,~\cite{ThomasAndCover}),
we introduce a random variable $J$, which is uniformly distributed
over $\{1,\cdots,n\}$, and continue as below.
\begin{eqnarray}
R_1 - \epsilon & \leq &  \epsilon_1
+ \frac{1}{n}\Bigg(\sum\limits_{i=1}^n I(W_1;Y_1(i)|\Ym_1^{i-1}, \tilde{\Ym}_e^{i+1}) {-}\:
\sum\limits_{i=1}^n I(W_1;Y_e(i)|\Ym_1^{i-1}, \tilde{\Ym}_e^{i+1}) \Bigg) \nonumber\\
& = & \epsilon_1
+ \frac{1}{n}\Bigg( \sum\limits_{i=1}^n I(V_1(i);Y_1(i)|U(i))
{-}\:
\sum\limits_{i=1}^n I(V_1(i);Y_e(i)|U(i)) \Bigg) \nonumber\\
& = & \epsilon_1
+ \sum\limits_{j=1}^n I(V_1(j);Y_1(j)|U(j))p(J=j)
{-}\:
\sum\limits_{j=1}^n I(V_1(j);Y_e(j)|U(j))p(J=j)\nonumber
\end{eqnarray}
\begin{eqnarray}
& \leq & \epsilon_1
+ \sum\limits_{j=1}^n I(V_1(j);Y_1(j)|V_2(j),U(j))p(J=j)
{-}\:
\sum\limits_{j=1}^n I(V_1(j);Y_e(j)|U(j))p(J=j) \nonumber\\
& = & \epsilon_1
+  I(V_1;Y_1|V_2,U)-I(V_1;Y_e|U),
\end{eqnarray}
where the last inequality follows from the fact that
$V_1(j)\to U(j)\to V_2(j)$, which implies
$$I(V_1(j);Y_1(j)|U(j)) \leq I(V_1(j);Y_1(j)|V_2(j),U(j)),$$
after using the fact that conditioning does not increase entropy,
and the last equality follows by using a standard information
theoretic argument in which we define random variables for
the single-letter expression, e.g., $V_1$ has the same
distribution as $V_1(J)$.
Hence, we obtain the bound,
\begin{eqnarray}
R_1 \leq \left[ I(V_1;Y_1|V_2,U)-I(V_1;Y_e|U) \right]^+,
\end{eqnarray}
for some auxiliary random variables that factors as
$p(u)p(v_1|u)p(v_2|u)p(x_1|v_1)p(x_2|v_2)p(y_1,y_2,y_e|x_1,x_2)$.

\begin{lemma}\label{thm:secrecylemma}
The secrecy constraint
$$R_1+R_2-\frac{1}{n} H(W_1,W_2|\Ym_e) \leq \epsilon$$
implies that
$$R_1-\frac{1}{n} H(W_1|\Ym_e) \leq \epsilon,$$
and
$$R_2-\frac{1}{n} H(W_2|\Ym_e) \leq \epsilon.$$
\end{lemma}

\begin{proof}
\begin{eqnarray}
\frac{1}{n}H(W_1|\Ym_e) & = & \frac{1}{n}
H(W_1,W_2|\Ym_e) - \frac{1}{n} H(W_2|\Ym_e,W_1) \nonumber\\
& \geq & R_1 - \epsilon + R_2 - \frac{1}{n} H(W_2|\Ym_e,W_1) \nonumber\\
& = & R_1 - \epsilon + \frac{1}{n} H(W_2) - \frac{1}{n} H(W_2|\Ym_e,W_1) \nonumber\\
& \geq & R_1 - \epsilon,
\end{eqnarray}
where the second equality follows by $H(W_2)=nR_2$, and the last
inequality follows due to the fact that conditioning does not
increase entropy. Second statement follows from a similar observation.
\end{proof}


\section{Proof of Theorem~\ref{thm:RO2}}
\label{sec:ProofOuterBound2}

From arguments given in~\cite[Lemma]{Costa:The87},
the assumed property of the channel implies the following.
\begin{equation}
I(\Vm_2;\Ym_2|\Vm_1) \leq I(\Vm_2;\Ym_1|\Vm_1)
\end{equation}

Then, by considering $V_1(i)=W_1$ and $V_2(i)=W_2$, for
$i=1,\cdots,n$, we get
\begin{equation}\label{eq:A.3.2}
I(W_2;\Ym_2|W_1) \leq I(W_2;\Ym_1|W_1)
\end{equation}

We continue as follows.
\begin{eqnarray}\label{eq:A.3.3}
\frac{1}{n}H(W_1,W_2|\Ym_1) & = & \frac{1}{n}H(W_1|\Ym_1)
+ \frac{1}{n}H(W_2|\Ym_1,W_1)\nonumber\\
& \leq & \epsilon_1 + \frac{1}{n}H(W_2|\Ym_1,W_1)\nonumber\\
& \leq & \epsilon_1 + \frac{1}{n}H(W_2|\Ym_2,W_1)\nonumber\\
& \leq & \epsilon_1 + \frac{1}{n}H(W_2|\Ym_2)\nonumber\\
& \leq & \epsilon_1 + \epsilon_2,
\end{eqnarray}
where the first inequality is due to the Fano's inequality
at the receiver $1$ with some
$\epsilon_1\to 0$ as $n\to\infty$,
the second inequality follows from \eqref{eq:A.3.2},
the third one is due to conditioning can not increase
entropy, and the last one follows from the
Fano's inequality at the receiver $2$ with some
$\epsilon_2\to 0$ as $n\to\infty$.

We then proceed following the standard techniques given
in~\cite{Wyner:The75,Csiszar:Broadcast78}. We first
state the following definitions.
\begin{eqnarray}
I_1 & \triangleq & \sum\limits_{i=1}^n
I(\tilde{\Ym}_e^{i+1};Y_1(i)|\Ym_1^{i-1})\\
\hat{I}_1 & \triangleq & \sum\limits_{i=1}^n
I(\Ym_1^{i-1};Y_e(i)|\tilde{\Ym}_e^{i+1})\\
I_3 & \triangleq & \sum\limits_{i=1}^n
I(\tilde{\Ym}_e^{i+1};Y_1(i)|\Ym_1^{i-1},W_1,W_2)\\
\hat{I}_3 & \triangleq & \sum\limits_{i=1}^n
I(\Ym_1^{i-1};Y_e(i)|\tilde{\Ym}_e^{i+1},W_1,W_2),
\end{eqnarray}
where $\tilde{\Ym}^{i+1} = [Y(i+1) \cdots Y(n)]$ for
random variable $Y$.

Then, we bound the sum rate as follows.
\begin{eqnarray}
R_1+R_2-\epsilon & \leq & \frac{1}{n}H(W_1,W_2|\Ym_e) \nonumber\\
& = & \frac{1}{n}\Bigg(H(W_1,W_2|\Ym_1) + I(W_1,W_2;\Ym_1) 
{-}\:
I(W_1,W_2;\Ym_e)\Bigg) \nonumber\\
& \leq & \epsilon_1+\epsilon_2
+ \frac{1}{n}\Bigg(\sum\limits_{i=1}^n I(W_1,W_2;Y_1(i)|\Ym_1^{i-1})
{-}\:
\sum\limits_{i=1}^n I(W_1,W_2;Y_e(i)|\tilde{\Ym}_e^{i+1}) \Bigg)\nonumber\\
& = &  \epsilon_1+\epsilon_2
+ \frac{1}{n}\bigg(
I_1 - I_3 - \hat{I}_1 + \hat{I}_3 
{+}\:
\sum\limits_{i=1}^n I(W_1,W_2;Y_1(i)|\Ym_1^{i-1}, \tilde{\Ym}_e^{i+1}) \nonumber\\
&&\quad
{-} \: \sum\limits_{i=1}^n I(W_1,W_2;Y_e(i)|\Ym_1^{i-1}, \tilde{\Ym}_e^{i+1})
\bigg) \nonumber\\
& = &  \epsilon_1+\epsilon_2
+ \frac{1}{n}\Bigg(\sum\limits_{i=1}^n I(W_1,W_2;Y_1(i)|\Ym_1^{i-1}, \tilde{\Ym}_e^{i+1}) {-}\:
\sum\limits_{i=1}^n I(W_1,W_2;Y_e(i)|\Ym_1^{i-1}, \tilde{\Ym}_e^{i+1}) \Bigg), \nonumber
\end{eqnarray}
where the first inequality is due to the secrecy requirement,
the last inequality follows by \eqref{eq:A.3.3}, and
the last equality follows by the fact that $I_1=\hat{I}_1$ and $I_3=\hat{I}_3$,
which can be shown using arguments similar to~\cite[Lemma $7$]{Csiszar:Broadcast78}.

We define $U(i)\triangleq(\Ym_1^{i-1},\tilde{\Ym}_e^{i+1},i)$,
$V_1(i)\triangleq(U(i),W_1)$, and $V_2(i)\triangleq(U(i),W_2)$.
Using standard techniques (see, e.g.,~\cite{ThomasAndCover}),
we introduce a random variable $J$, which is uniformly distributed
over $\{1,\cdots,n\}$, and continue as below.

$\frac{1}{n}\Bigg(\sum\limits_{i=1}^n I(W_1,W_2;Y_1(i)|\Ym_1^{i-1}, \tilde{\Ym}_e^{i+1})
{-} \sum\limits_{i=1}^n I(W_1,W_2;Y_e(i)|\Ym_1^{i-1}, \tilde{\Ym}_e^{i+1}) \Bigg)$
\begin{eqnarray}
&=& \frac{1}{n}\Bigg( \sum\limits_{i=1}^n I(V_1(i),V_2(i);Y_1(i)|U(i))
{-}\: \sum\limits_{i=1}^n I(V_1(i),V_2(i);Y_e(i)|U(i)) \Bigg) \nonumber\\
&=&  \sum\limits_{j=1}^n I(V_1(j),V_2(j);Y_1(j)|U(j))p(J=j)
{-}\:   \sum\limits_{j=1}^n I(V_1(j),V_2(j);Y_e(j)|U(j))p(J=j) \nonumber\\
&=&  I(V_1,V_2;Y_1|U)
-I(V_1,V_2;Y_e|U),
\end{eqnarray}
where, using a standard information theoretic argument, we have
defined the random variables for the single-letter expression,
e.g., $V_1$ has the same distribution as $V_1(J)$.
Now,
due to the memoryless property of the channel, we have
$(U,V_1,V_2) \to (X_1,X_2) \to (Y_1,Y_2,Y_e)$, which implies
$p(y_1,y_2,y_e|x_1,x_2,v_1,v_2,u)=p(y_1,y_2,y_e|x_1,x_2)$.
As we define $V_1(J)=(U(J),W_1)$ and $V_2(J)=(U(J),W_2)$,
we have $V_1 \to U \to V_2$, which implies
$p(v_1,v_2|u)=p(v_1|u)p(v_2|u)$. Finally, as
$X_1(J)$ is a stochastic function of $W_1$,
$X_2(J)$ is a stochastic function of $W_2$, and
$W_1$ and $W_2$ are independent, we have
$X_1(J) \to V_1(J) \to V_2(J)$,
$X_2(J) \to V_2(J) \to V_1(J)$, and
$X_1(J) \to (V_1(J),V_2(J)) \to X_2(J)$, which together implies
that
$p(x_1,x_2|v_1,v_2,u)=p(x_1,x_2|v_1,v_2)
=p(x_1|v_1,v_2)p(x_2|v_1,v_2)=p(x_1|v_1)p(x_2|v_2)$.

Using this in the above equation, we obtain a sum-rate bound,
\begin{eqnarray}
R_1+R_2 \leq \left[ I(V_1,V_2;Y_1|U)-I(V_1,V_2;Y_e|U) \right]^+,
\end{eqnarray}
for some auxiliary random variables that factors as
$p(u)p(v_1|u)p(v_2|u)p(x_1|v_1)p(x_2|v_2)p(y_1,y_2,y_e|x_1,x_2)$,
if (\ref{eq:ProofOutCond1}) holds.


\section{Proof of Corollary~\ref{thm:C1}}
\label{sec:ProofC1}

Achievability follows from Theorem~\ref{thm:R} by only utilizing
$S_1$ and $O_2$ together with the second equation
in \eqref{eq:Speceq1}, where we set $R_2=0$ and set $R_1$ as follows.
For a given $p\in \Pc(S_1,O_2)$, if
$I(S_1;Y_1|O_2,Q) \leq I(S_1;Y_e|Q)$, we set $R_1=0$;
otherwise we assign the following rates.
\begin{eqnarray}
R_{S_1} & = & I(S_1;Y_1|O_2,Q)-I(S_1;Y_e|Q) \nonumber\\
R_{S_1}^x & = & I(S_1;Y_e|Q) \nonumber\\
R_{O_2}^x & = & I(O_2;Y_e|S_1,Q), \nonumber
\end{eqnarray}
where $R_1=R_{S_1}$.

Converse follows from Theorem~\ref{thm:RO1}. That is, if
$(R_1,R_2)$ is achieavable, then
$R_1\leq \max\limits_{p\in\Pc_O} \: I(S_1;Y_1|O_2,Q)-I(S_1;Y_e|Q)$
and $R_2=0$, due to the first condition
given in \eqref{eq:Speceq1}.


\section{Proof of Corollary~\ref{thm:C2}}
\label{sec:ProofC2}

Achievability follows from Theorem~\ref{thm:R} by only utilizing
$S_1$ and $C_2$ together with the channel condition given in~\eqref{eq:Speceq2}.
For a given $p\in \Pc(S_1,C_2)$, if
$I(S_1,C_2;Y_1|Q) \leq I(S_1,C_2;Y_e|Q)$, we set $R_1=R_2=0$;
otherwise we assign the following rates.
\begin{eqnarray}
R_{S_1} & = & I(S_1;Y_1|C_2,Q)-I(S_1;Y_e|C_2,Q) \nonumber\\
R_{S_1}^x & = & I(S_1;Y_e|C_2,Q) \nonumber\\
R_{C_2} & = & I(C_2;Y_1|Q)-I(C_2;Y_e|Q) \nonumber\\
R_{C_2}^x & = & I(C_2;Y_e|Q), \nonumber
\end{eqnarray}
where $R_1=R_{S_1}$ and $R_2=R_{C_2}$.

Converse follows from Theorem~\ref{thm:RO2}
as the needed condition is satisfied by the channel.


\section{Proof of Corollary~\ref{thm:C3}}
\label{sec:ProofC3}

Achievability follows from Theorem~\ref{thm:R} by only utilizing
$S_1$ and $O_2$ together with the channel condition given in~\eqref{eq:Speceq3}.
For a given $p\in \Pc(S_1,O_2)$, if
$I(S_1,O_2;Y_1|Q) \leq I(S_1,O_2;Y_e|Q)$, we set $R_1=R_2=0$;
otherwise we assign the following rates.
\begin{eqnarray}
R_{S_1} & = & I(S_1,O_2;Y_1|Q)-I(S_1,O_2;Y_e|Q) \nonumber\\
R_{S_1}^x & = & I(S_1,O_2;Y_e|Q)-I(O_2;Y_1|S_1,Q) \nonumber\\
R_{O_2}^x & = & I(O_2;Y_1|S_1,Q), \nonumber
\end{eqnarray}
where $R_1=R_{S_1}$ and $R_2=0$.

Converse follows from Theorem~\ref{thm:RO2}
as the needed condition is satisfied by the channel.


\section{Proof of Corollary~\ref{thm:C4}}
\label{sec:ProofC4}

Achievability follows from Theorem~\ref{thm:R} by only utilizing
$S_1$ and $O_2$ together with the channel condition given in~\eqref{eq:Speceq4}.
For a given $p\in \Pc(S_1,O_2)$, if
$I(S_1,O_2;Y_1|Q) \leq I(S_1,O_2;Y_e|Q)$, we set $R_1=R_2=0$;
otherwise we assign the following rates.
\begin{eqnarray}
R_{S_1} & = & I(S_1,O_2;Y_1|Q)-I(S_1,O_2;Y_e|Q) \nonumber\\
R_{S_1}^x & = & I(S_1;Y_e|Q) \nonumber\\
R_{O_2}^x & = & I(O_2;Y_e|S_1,Q), \nonumber
\end{eqnarray}
where $R_1=R_{S_1}$ and $R_2=0$.

Converse follows from Theorem~\ref{thm:RO2}
as the needed condition is satisfied by the channel.


\section{Proof of Corollary~\ref{thm:C6}}
\label{sec:ProofC6}

For a given MAC-E with $p(y_1,y_e|x_1,x_2)$,
we consider an IC-E defined by $$p(y_1,y_2,y_e|x_1,x_2)=p(y_1,y_e|x_1,x_2)\delta(y_2-y_1)$$
and utilize Theorem~\ref{thm:R} with $p\in\Pc(C_1,C_2)$ satisfying \eqref{eq:Speceq7}.
Then, the achievable region becomes
\begin{eqnarray}
R_1=R_{C_1}&\leq&I(C_1;Y_1|C_2,Q)-R_{C_1}^x \nonumber\\
R_2=R_{C_2}&\leq&I(C_2;Y_1|C_1,Q)+R_{C_1}^x 
- I(C_1,C_2;Y_e|Q), \nonumber\\
R_1+R_2=R_{C_1}+R_{C_2} &\leq& I(C_1,C_2;Y_1|Q)
- I(C_1,C_2;Y_e|Q),
\end{eqnarray}
where $I(C_1;Y_e|Q)\leq R_{C_1}^x \leq I(C_1;Y_e|C_2,Q)$
and $R_{C_2}^x = I(C_1,C_2;Y_e|Q) - R_{C_1}^x$ . Hence,
$R_1+R_2=[I(C_1,C_2;Y_1|Q)- I(C_1,C_2;Y_e|Q)]^+$ is achievable for
any $p\in\Pc(C_1,C_2)$ satisfying \eqref{eq:Speceq7}.

The following outer bound on the sum rate follows by
Theorem~\ref{thm:RO2}, as the constructed IC-E satisfies
the needed condition of the theorem.
$$R_1+R_2\leq I(C_1,C_2;Y_1|Q)-I(C_1,C_2;Y_e|Q),$$
for any $p\in\Pc(C_1,C_2)$. Which is what needed to be shown.


\section{Proof of Corollary~\ref{thm:C7}}
\label{sec:ProofC7}

We first remark that
$R^{[\textrm{NF}]}$
will remain the same if we restrict the union over the set of
probability distributions
$$\tilde{\Pc}(S_1,O_2,|\Qc|=1)\triangleq
\{p \: | \: p\in\Pc(S_1,O_2,|\Qc|=1),
I(O_2;Y_e)\leq I(O_2;Y_1)\}.$$

As for any $p\in\Pc(S_1,O_2,|\Qc|=1)$ satisfying $I(O_2;Y_e) \geq I(O_2;Y_1)$,
$R_1(p)= [I(S_1;Y_1|O_2)-I(S_1;Y_e|O_2)]^+$ since
$I(O_2;Y_1) \leq I(O_2;Y_e) \leq I(O_2;Y_e|S_1)$ in this case.
And the highest rate achievable with the NF scheme occurs
if $O_2$ is chosen to be deterministic, and hence $I(O_2;Y_1) = I(O_2;Y_e)$
case will result in the highest rate among the probability distributions
$p\in\Pc(S_1,O_2,|\Qc|=1)$ satisfying $I(O_2;Y_e) \geq I(O_2;Y_1)$.
Therefore, without loss of generality, we can write
$$R^{\textrm{[NF]}}=
\max\limits_{p\in\tilde{\Pc}(S_1,O_2,|\Qc|=1)}
\: R_1(p),$$
where $R_1(p)=[I(S_1;Y_1|O_2)+\min\{I(O_2;Y_1),I(O_2;Y_e|S_1)\}
-I(S_1,O_2;Y_e)]^+$.

Now, fix some $p\in\tilde{\Pc}(S_1,O_2,|\Qc|=1)$, and set
$R_{O_2}^x=\min\{I(O_2;Y_1),I(O_2;Y_e|S_1)\}$,
$R_{S_1}^x=I(S_1,O_2;Y_e)-\min\{I(O_2;Y_1),I(O_2;Y_e|S_1)\}$, and
$R_{S_1}=I(S_1;Y_1|O_2)-I(S_1,O_2;Y_e)+\min\{I(O_2;Y_1),I(O_2;Y_e|S_1)\}$,
where we set $R_1=0$ if the latter is negative.
Here,
\begin{eqnarray}\label{eq:R_NF1}
R_1&=&\bigg[I(S_1;Y_1|O_2)+\min\{I(O_2;Y_1),I(O_2;Y_e|S_1)\} 
- I(S_1,O_2;Y_e)\bigg]^+ \nonumber
\end{eqnarray}
is achievable, i.e., $(R_1(p),0)\in\Rc^{\textrm{RE}}$
for any $p\in\tilde{\Pc}(S_1,O_2,|\Qc|=1)$.
Observing that $(R^{\textrm{[NF]}},0)\in\Rc^{\textrm{RE}}$,
we conclude that the noise forwarding (NF) scheme
of~\cite{Lai:The08} is a special case of the proposed scheme.


\section{Proof of Corollary~\ref{thm:C8}}
\label{sec:ProofC8}

For any given $p\in\Pc(S_1,O_2,|\Qc|=1)$ satisfying
$I(O_2;Y_e) \leq I(O_2;Y_1)$, we see that
$R_1=[I(S_1,O_2;Y_1)-I(S_1,O_2;Y_e)]^+$ is achievable
due to \eqref{eq:Speceq8} and Corollary~\ref{thm:C7}.
The converse follows by Theorem~\ref{thm:RO2}
as the needed condition is satisfied by considering an
IC-E defined as $p(y_1,y_e|x_1,x_2)\delta(y_2-y_1)$,
where we set $|\Qc|=1$ in the upper bound
as the time sharing random variable is not needed for
this scenario, and further limit the input distributions
to satisfy $I(O_2;Y_e) \leq I(O_2;Y_1)$. The latter does not
reduce the maximization for the upper bound due to
a similar reasoning given in the Proof of
Corollary~\ref{thm:C7}.


\bibliographystyle{IEEEtran}

\begin{thebibliography}{10}


\bibitem{Koyluoglu:Onthe08}
O.~O.~Koyluoglu and H.~El~Gamal,
``On the Secrecy Rate Region for the Interference Channel,''
in {\em Proc. 2008 IEEE International Symposium on Personal, Indoor and Mobile Radio Communications ({PIMRC}'08)},
Cannes, France, Sep. 2008.

\bibitem{Han:A81}
T.~Han and K.~Kobayashi,
``A new achievable rate region for the interference channel,''
{\em IEEE Trans. Inf. Theory},
vol.~27, no.~1, pp.~49--60, Jan. 1981.

\bibitem{Chong:On08}
H.~Chong, M.~Motani, H.~Garg, and H.~El~Gamal,
``On the {Han-Kobayashi} region for the interference channel,''
{\em IEEE Trans. Inf. Theory},
vol.~54, pp.~3188--3195, Jul. 2008.

\bibitem{Sato:The81}
H.~Sato,
``The capacity of the gaussian interference channel under strong interference,''
{\em IEEE Trans. Inf. Theory},
vol.~27, no.~6, pp.~786--788, Nov. 1981.

\bibitem{Costa:The87}
M.~H.~Costa and A.~El~Gamal,
``The capacity region of the discrete memoryless interference channel with strong interference,''
{\em IEEE Trans. Inf. Theory},
vol.~33, no.~5, pp.~710--711, Sep. 1987.

\bibitem{Shang:A}
X.~Shang, G.~Kramer, and B.~Chen,
``A new outer bound and the noisy-interference sum-rate capacity for gaussian interference channels,''
{\em IEEE Trans. Inf. Theory},
vol.~55, no.~2, pp.~689--699, Feb. 2009.

\bibitem{Annapureddy:Gaussian}
V.~S.~Annapureddy and V.~V.~Veeravalli,
``Gaussian interference networks: Sum capacity in the low interference regime and new outer bounds on the capacity region,''
{\em IEEE Trans. Inf. Theory},
vol.~55, no.~7, pp.~3032--3050, Jul. 2009.

\bibitem{Motahari:Capacity}
A.~S.~Motahari and A.~K.~Khandani,
``Capacity bounds for the gaussian interference channel,''
{\em IEEE Trans. Inf. Theory},
vol.~55, no.~2, pp.~620--643, Feb. 2009.

\bibitem{Liang:Cognitive07}
Y.~Liang, A.~Somekh-Baruch, H.~V.~Poor, S.~Shamai, and S.~Verdu,
``Cognitive interference channels with confidential messages,''
in {\em Proc. of the 45th Annual Allerton Conference on Communication, Control and Computing},
Monticello, IL, Sep. 2007.

\bibitem{Liu:Discrete08}
R.~Liu, I.~Maric, P.~Spasojevic, and R.~D.~Yates,
``Discrete memoryless interference and broadcast channels with confidential messages:
Secrecy rate regions,''
{\em IEEE Trans. Inf. Theory},
vol.~54, no.~6, pp.~2493--2507, Jun. 2008.

\bibitem{Koyluoglu:On08}
O.~O.~Koyluoglu, H.~El~Gamal, L.~Lai, and H.~V.~Poor,
``On the secure degrees of freedom in the K-user gaussian interference channel,''
in {\em Proc. 2008 IEEE International Symposium on Information Theory (ISIT'08)},
Toronto, ON, Canada, Jul. 2008.

\bibitem{Koyluoglu:Interference}
O.~O.~Koyluoglu, H.~El~Gamal, L.~Lai, and H.~V.~Poor,
``Interference alignment for secrecy,''
{\em IEEE Trans. Inf. Theory},
to appear.

\bibitem{Wyner:The75}
A.~D.~Wyner,
``The Wire-Tap Channel,''
{\em The Bell System Technical Journal},
vol.~54, no.~8, pp.~1355-–1387, Oct. 1975.

\bibitem{Csiszar:Broadcast78}
I.~Csisz\'ar and J.~K\"orner,
``Broadcast channels with confidential messages,''
{\em IEEE Trans. Inf. Theory},
vol.~24, no.~3, pp.~339-–348, May. 1978.

\bibitem{Negi:Secret05}
R.~Negi and S.~Goel,
``Secret communication using artificial noise,''
in {\em Proc. 2005 IEEE 62nd Vehicular Technology Conference, (VTC-2005-Fall)},
2005.

\bibitem{Tekin:The08}
E.~Tekin and A.~Yener,
``The general Gaussian multiple access and two-way wire-tap channels: Achievable rates and cooperative jamming,''
{\em IEEE Trans. Inf. Theory},
vol.~54, no.~6, pp.~2735–-2751, Jun. 2008.

\bibitem{ThomasAndCover}
T.~Cover and J.~Thomas,
``Elements of Information Theory,''
John Wiley and Sons, Inc., 1991.

\bibitem{Lai:The08}
L.~Lai and H.~El~Gamal,
``The relay-eavesdropper channel: Cooperation for secrecy,''
{\em IEEE Trans. Inf. Theory},
vol.~54, no.~9, pp.~4005-–4019, Sep. 2008.

\bibitem{Carleial:Interference78}
A.~Carleial,
``Interference channels,''
{\em IEEE Trans. Inf. Theory},
vol.~24, no.~1, pp.~60-–70, Jan. 1978.

\bibitem{Leung-Yan-Cheong:The78}
S.~Leung-Yan-Cheong and M.~Hellman,
``The gaussian wire-tap channel,''
{\em IEEE Trans. Inf. Theory},
vol.~24, no.~4, pp.~451-–456, Jul. 1978.

\end{thebibliography}

\end{document}